\pdfoutput=1
\documentclass[11pt,a4paper]{article}%
\usepackage{amsmath}
\usepackage{amsfonts}
\usepackage{amssymb}
\usepackage{graphicx}
\usepackage[onehalfspacing]{setspace}
\providecommand{\U}[1]{\protect\rule{.1in}{.1in}}
\newtheorem{theorem}{Theorem}

\newenvironment{proof}[1][Proof]{\noindent\textbf{#1.} }{\ \rule{0.5em}{0.5em}}
\begin{document}

\title{An efficient asymptotic approach for testing monotone proportions assuming an
underlying logit based order dose-response model}
\author{Nirian Martin$^{1}$ and Raquel Mata$^{2}$\\$^{1}${\small Department of Statistics, Carlos III University of Madrid, 28903
Getafe (Madrid), Spain.}\\$^{2}${\small Department of Statistics and O.R., Complutense University of
Madrid, 28040 Madrid, Spain.}}
\date{}
\maketitle

\begin{abstract}
When an underlying logit based order dose-response model is considered with
small or moderate sample sizes, the Cochran-Armitage (CA) test represents the
most efficient test in the framework of the test-statistics applied with
asymptotic distributions for testing monotone proportions. The Wald and
likelihood ratio (LR) test have much worse behaviour in type error I in
comparison with the CA test. It suffers, however, from the weakness of not
maintaining the nominal size. In this paper a family of test-statistics based
on $\phi$-divergence measures is proposed and their asymptotic distribution
under the null hypothesis is obtained either for one-sided or two-sided
hypothesis testing. A numerical example based on real data illustrates that
the proposed test-statistics are simple for computation and moreover, the
necessary goodness-of-fit test-statistic are easily calculated from them. The
simulation study shows that the test based on the Cressie and Read
(\textit{Journal of the Royal Statistical Society}, \textit{Series B},
\textbf{46, }440-464, 1989) divergence measure usually provides a better
nominal size than the CA test for small and moderate sample sizes.

\end{abstract}

\textbf{Keywords}: $I\times2$ contingency table, order-restricted inference,
dose-response logit model, Cochran-Armitage test, phi-divergence test statistic

\section{Introduction\label{Sec1}}

In many applications, it is natural to predict that the relationship between
two variables satisfies a rather vague condition such as `$Y$ tends to
increase as $X$ increases'. For instance, in many clinical or epidemiological
studies, an important objective is to asses the existence of a monotonic
dose-response relationship between a disease and an ordered exposure, that is
a relationship in which disease risk increases with each increment of
exposure. A common way for a researcher to handle this, is to construct a
generalized linear model with binary data in which $X$ (dose) has a linear
effect on some scale, on a response variable $Y.$ For a binary response $Y$,
we denote by $\pi(x)=\Pr(Y=1|X=x)$ the probability of a success given a dose
$x$, the unknown values for which we desire to make decisions. If we consider
$I$ doses $0<x_{1}<x_{2}<....<x_{I}$ and each of them is given to $n_{i}$
individuals, $i=1,...,I$ respectively, we have $I$ independent binomial random
variables $N_{i1}|X=x_{i}\sim\mathcal{B}(n_{i},\pi(x_{i}))$,$\;i=1,...,I$,
representing the number of successes out of $n_{i}$ trials when the level of
the predictor, the dose, is $x_{i}$, $i=1,...,I$. The information of interest
when we have a realization in a sample can be summarized as%
\[%
\begin{tabular}
[c]{||c|c||c|c||}\hline\hline
$x_{1}$ & $n_{1}$ & $n_{11}$ & $n_{12}=n_{1}-n_{11}$\\
$\vdots$ & $\vdots$ & $\vdots$ & $\vdots$\\
$x_{i}$ & $n_{i}$ & $n_{i1}$ & $n_{i2}=n_{i}-n_{i1}$\\
$\vdots$ & $\vdots$ & $\vdots$ & $\vdots$\\
$x_{I}$ & $n_{I}$ & $n_{I1}$ & $n_{I2}=n_{I}-n_{I1}$\\\hline\hline
\end{tabular}
\ \ \ .
\]
Note that we have an $I\times2$ contingency table, expressed in vector
notation by%
\[
\boldsymbol{N}=(\boldsymbol{N}_{1}^{T},...,\boldsymbol{N}_{i}^{T}%
,...,\boldsymbol{N}_{I}^{T})^{T}%
\]
where $\boldsymbol{N}_{i}=(N_{i1},N_{i2})^{T}$, with $N_{i2}=n_{i}-N_{i1}%
$\ being the number of failures out of $n_{i}$ trials, $i=1,...,I$. As we are
dealing with a product binomial sample or a multi-sample of binomial random
variables, we have $\boldsymbol{N}_{i}^{T}\boldsymbol{1}_{2}=n_{i}$ and
$\boldsymbol{N}^{T}\boldsymbol{1}_{I2}=n$, where $n\equiv%
{\textstyle\sum\nolimits_{i=1}^{I}}
n_{i}$ and $n_{i}$, $i=1,...,I$ are prefixed known values.

The statistical problem consisting in testing the equality of $I$ binomial
proportions against a monotone trend in proportions at the same or opposite
direction of the doses has been extensively studied in different research
settings. One of the most frequently used test-statistic is, by far, the
Cochran-Armitage (CA) test, defined as%
\begin{equation}
T_{n,CA}=\frac{%
{\displaystyle\sum\limits_{i=1}^{I}}
N_{i1}(x_{i}-\bar{x})}{\left(  \widehat{p}_{\bullet1}(1-\widehat{p}_{\bullet
1})%
{\displaystyle\sum\limits_{i=1}^{I}}
N_{i1}(x_{i}-\bar{x})^{2}\right)  ^{\frac{1}{2}}}, \label{ca}%
\end{equation}
where $\bar{x}=\frac{1}{n}%
{\textstyle\sum\nolimits_{i=1}^{I}}
n_{i}x_{i}$ and $\widehat{p}_{\bullet1}=(\sum_{i=1}^{I}N_{i1})/n$. It was
introduced by Cochran \cite{1} and Armitage \cite{2}, and discussed in Mantel
\cite{3} as special case of the extended Mantel-Haenszel test for several
$I\times2$ contingency tables, each one corresponding to a stratum or
categories of a confounding variable. It can be found expressed in several
ways but (\ref{ca}) corresponds with the one given in Tarone and Gart
\cite{4}, at the end of Section 2. It assumes that parameter $\pi(x)$ is
linked to the linear predictor
\begin{equation}
\eta=g(\pi)=\alpha+\beta x, \label{link}%
\end{equation}
where the link function, $g$, is a monotone and twice differentiable function
over the interval $[x_{1},x_{I}]$. The square of the Cohran-armitage test is a
score test-statistic (Rao, \cite{5}), where under $H_{0}:\beta=0$ it requires
to replace the nuisance parameter $\alpha$ by its maximum likelihood estimator
(MLE), and in comparison with other test-statistic focused on the same model
assumption, such as Wald and likelihood ratio (LR) tests, it does not depend
on the functional shape of function $g$. Taking into account such a property,
Cox \cite[page 65]{6} considered that it is a kind of nonparametric
test-statistic. In Cox \cite{6b} and Mantel \cite{3} the logit function,%
\begin{equation}
\eta=g(\pi)=\log\left(  \frac{\pi}{1-\pi}\right)  , \label{logit}%
\end{equation}
was applied as link function and in Tarone and Gart \cite{4} was found it as
an optimal function in terms of the Pitman asymptotic relative efficiency.

In the existing literature on dose-response models we can distinguish model
based techniques (parametric procedure) and isotonic regression or
order-restricted techniques (non-parametric procedure). See Barlow et al.
\cite{6c}, Robertson et al. \cite{6d} or Silvapulle and Sen \cite{6e}\ for
more detailed information about both types of procedures. Leuraud and Benichou
\cite{7} made comparison studies of type I error and power for both kind of
test-statistics (CA test and isotonic regression among others) for small and
moderate sample sizes and their conclusion is very similar to the one given in
Agresti and Coull \cite{8}, for LR tests, logit model based one and the
order-restricted one: the model based test is good in type I error and power
properties but the researcher must be cautious in checking the model
assumptions previously, i.e. an additional goodness-of-fit test is needed for
the linear logit model. The aforementioned methods are based on asymptotic
distributions of the test-statistics. In Hirji and Tang \cite{9}, Tang et al.
\cite{10} and Shan et al. \cite{11} exact methods were proposed and they solve
an important weakness associated with the usually applied asymptotic methods:
for small and moderate sample sizes the nominal size of the test is not
usually preserved. That is, the exact significance level tends to exceed the
nominal level, by a big margin in the case of the Wald and LR test-statistic.
Such a problem was theoretically studied in Kang and Lee \cite{12} for the
two-sided CA test. Based on the logit link function, our interest in this
paper is to find a new family of test statistics with the same asymptotic
distribution as the LR test (see Agresti and Coull \cite{8}) which correct the
weakness in the preservation of the nominal size and maintain similar
properties in power. The CA test-statistic is useful as guideline for
comparison, since it has the best behavior between the asymptotic test-statistics.

This article is organized as follows. In Section \ref{Sec2} the proposed
test-statistics are presented and their asymptotic distribution is found for
one-sided and two-sided alternatives. Section \ref{Ex} is devoted to
illustrate the method with a real data example and\ in Section \ref{MC} the
performance in error I and power of the proposed test-statistics is studied
and compared with the CA test.

\section{Proposed test-statistics\label{Sec2}}

Under the model assumption (\ref{logit}), the conditional probability vector
of $\boldsymbol{N}$ is given by%
\[
\boldsymbol{\pi}(\alpha,\beta)\boldsymbol{=}\left(  \pi_{11}(\alpha,\beta
),\pi_{12}(\alpha,\beta),...,\pi_{I1}(\alpha,\beta),\pi_{I2}(\alpha
,\beta)\right)  ^{T},
\]
where%
\[
\pi_{ij}(\alpha,\beta)=\left\{
\begin{array}
[c]{ll}%
1-\pi(x_{i}), & j=1\\
\pi(x_{i}), & j=2
\end{array}
\right.  ,
\]
and the joint probability vector of $\boldsymbol{N}$%
\begin{equation}
\boldsymbol{p}(\alpha,\beta)\boldsymbol{=}\left(  p_{11}(\alpha,\beta
),p_{12}(\alpha,\beta),...,p_{I1}(\alpha,\beta),p_{I2}(\alpha,\beta)\right)
^{T}, \label{jointPr}%
\end{equation}
where%
\[
p_{i1}(\alpha,\beta)=(X=x_{i},Y=j)=\Pr(Y=j|X=x_{i})\Pr(X=x_{i})=\pi
_{ij}(\alpha,\beta)\frac{n_{i}}{n}=\left\{
\begin{array}
[c]{ll}%
\frac{n_{i}}{n}(1-\pi(x_{i})), & j=1\\
\frac{n_{i}}{n}\pi(x_{i}), & j=2
\end{array}
\right.  .
\]

We shall test the null hypothesis of no relationship between the binary
response $Y$ and an ordered categorical explanatory variable $X$ (doses)
against the one-sided alternative hypothesis $H_{1}$ of an increasing
dose-response relationship between a response variable $Y$ and (doses) $X$%
\begin{subequations}
\begin{align}
&  H_{0}:\;\pi(x_{1})\leq\pi(x_{2})\leq\cdots\leq\pi(x_{I-1})\leq\pi
(x_{I})\text{,}\label{test1a}\\
&  H_{1}:\;\pi(x_{1})>\pi(x_{2})>\cdots>\pi(x_{I-1})>\pi(x_{I})\text{.}
\label{test1b}%
\end{align}
Taking into account%
\end{subequations}
\[
\beta>0\Longleftrightarrow\beta\left(  x_{i+1}-x_{i}\right)
>0\Longleftrightarrow\underset{=\pi(x_{i})}{\underbrace{\frac{\exp\left(
\alpha+\beta x_{i+1}\right)  }{1+\exp\left(  \alpha+\beta x_{i+1}\right)  }}%
}>\underset{=\pi(x_{i+1})}{\underbrace{\frac{\exp\left(  \alpha+\beta
x_{i}\right)  }{1+\exp\left(  \alpha+\beta x_{i}\right)  }}},
\]
we can see that (\ref{test1a})-(\ref{test1b}) is equivalent to%
\begin{equation}
H_{0}:\beta\leq0\text{ versus }H_{1}:\beta>0. \label{test1}%
\end{equation}
It is important to mention that sometime (\ref{test1a}) and $H_{0}$\ in
(\ref{test1}) are expressed with equalities (see for instance Shan et al.
\cite[Section 2]{11}), but the procedure used for the test-statistic is
equivalent since the shape and the asymptotic distribution of the
test-statistic is the same. We prefer using this shape since in order to
justify later the goodness of fit test-statistic is more coherent.

We shall also consider the two-sided alternative hypothesis $H_{1}^{\prime}$
of a decreasing or increasing dose-response relationship between a response
variable $Y$ and (doses) $X$,%
\begin{subequations}
\begin{align}
&  H_{0}^{\prime}:\;\pi(x_{1})=\pi(x_{2})=\cdots=\pi(x_{I-1})=\pi
(x_{I})\text{,}\label{test2a}\\
&  H_{1}^{\prime}:\;\left(  \pi(x_{1})\leq\pi(x_{2})\leq\cdots\leq\pi
(x_{I-1})\leq\pi(x_{I})\text{ and }\pi(x_{1})<\pi(x_{I})\right)
\label{test2b}\\
&  \qquad\text{or }\left(  \pi(x_{1})\geq\pi(x_{2})\geq\cdots\geq\pi
(x_{I-1})\geq\pi(x_{I})\text{ and }\pi(x_{1})>\pi(x_{I})\right)
\text{,}\nonumber
\end{align}
which is equivalent to%
\end{subequations}
\begin{equation}
H_{0}^{\prime}:\beta=0\text{ versus }H_{1}^{\prime}:\beta\neq0. \label{test2}%
\end{equation}
The asymptotic distribution of the CA test-statistics, (\ref{ca}), under
$H_{0}$ in (\ref{test1}) and under $H_{0}^{\prime}$\ in (\ref{test2}), is
standard normal. In practice, we shall prefer use the chi-square distribution
with 1 degree of freedom ($\chi_{1}^{2}$) for $T_{n,CA}^{2}$ when we follow
the two-sided test.

Let $(\widehat{\alpha},\widehat{\beta})$ be the MLE of parameters in the
linear logit model (\ref{logit}) and $(\widetilde{\alpha},\widetilde{\beta})$
the MLE in the linear logit model (\ref{logit}) with restriction $\beta\leq0$.
If $\widehat{\beta}>0$, the LR test-statistic for the one-sided test
(\ref{test1}) is given by%
\[
G_{n}^{2}=2\left(  \sum_{i=1}^{I}N_{i1}\log\left(  \frac{\pi_{i1}%
(\widehat{\alpha},\widehat{\beta})}{\pi_{i1}(\widetilde{\alpha}%
,\widetilde{\beta})}\right)  +\sum_{i=1}^{I}N_{i2}\log\left(  \frac{\pi
_{i2}(\widehat{\alpha},\widehat{\beta})}{\pi_{i2}(\widetilde{\alpha
},\widetilde{\beta})}\right)  \right)  .
\]
If $\widehat{\beta}\leq0$, then $\widehat{\beta}=\widetilde{\beta}$ and the LR
test-statistic for the one-sided test (\ref{test1}) is given by $G_{n}^{2}=0$,
which means that the null hypothesis (lack of positive monotonicity) is always
accepted. In such a case, we should perform the opposite test $H_{0}:\beta
\geq0$ versus $H_{1}:\beta<0$, in order to demonstrate negative monotonicity.
If $\widehat{\beta}>0$, then $\widetilde{\beta}=0$ and the LR test-statistic
for the one-sided test (\ref{test1}) is given by%
\begin{equation}
G_{n}^{2}=2\left(  \sum_{i=1}^{I}N_{i1}\log\left(  \frac{\pi_{i1}%
(\widehat{\alpha},\widehat{\beta})}{\pi_{i1}(\widetilde{\alpha},0)}\right)
+\sum_{i=1}^{I}N_{i2}\log\left(  \frac{\pi_{i2}(\widehat{\alpha}%
,\widehat{\beta})}{\pi_{i2}(\widetilde{\alpha},0)}\right)  \right)  ,
\label{LRT}%
\end{equation}
where $\pi_{i1}(\widetilde{\alpha},0)=\widehat{p}_{\bullet1}=(\sum_{i=1}%
^{I}N_{i1})/n$, $\pi_{i2}(\widetilde{\alpha},0)=1-\pi_{i1}(\widetilde{\alpha
},0)$, $\pi_{i1}(\widehat{\alpha},\widehat{\beta})=\frac{\exp\{\widehat{\alpha
}+\widehat{\beta}x_{i}\}}{1+\exp\{\widehat{\alpha}+\widehat{\beta}x_{i}\}}$
and $\pi_{i2}(\widehat{\alpha},\widehat{\beta})=1-\pi_{i1}(\widehat{\alpha
},\widehat{\beta})$. The asymptotic distribution of the LR test-statistic for
(\ref{test1}), as $n$ goes to infinite, is chi-bar square with two summands
(see Agresti and Coull \cite{8} for more details). The LR test-statistic for
the two-sided test (\ref{test2}) is also given by (\ref{LRT}) and its
asymptotic distribution is $\chi_{1}^{2}$.

Now we are going to construct a new family of test-statistics inspired in that
(\ref{LRT}) can be expressed in terms of the Kullback divergence measure
between the empirical and model joint probability vectors, as follows%
\begin{equation}
G_{n}^{2}=2\left(  \mathrm{d}_{Kull}(\widehat{\boldsymbol{p}},\boldsymbol{p}%
(\widetilde{\alpha},0))-\mathrm{d}_{Kull}(\widehat{\boldsymbol{p}%
},\boldsymbol{p}(\widehat{\alpha},\widehat{\beta}))\right)  , \label{LRT2}%
\end{equation}
where $\widehat{\boldsymbol{p}}$ is the empirical joint probability vector of
$\boldsymbol{N}$, $\widehat{\boldsymbol{p}}=\frac{\boldsymbol{N}}{n}$, i.e.
$\widehat{\boldsymbol{p}}\boldsymbol{=}\left(  \widehat{p}_{11},\widehat{p}%
_{12},...,\widehat{p}_{I1},\widehat{p}_{I2}\right)  ^{T}$, with $p_{ij}%
=\frac{N_{ij}}{n}$, $\boldsymbol{p}(\widehat{\alpha},\widehat{\beta
})\boldsymbol{=}\left(  p_{11}(\widehat{\alpha},\widehat{\beta}),p_{12}%
(\widehat{\alpha},\widehat{\beta}),...,p_{I1}(\widehat{\alpha},\widehat{\beta
}),p_{I2}(\widehat{\alpha},\widehat{\beta})\right)  ^{T}$ with $p_{i1}%
(\widehat{\alpha},\widehat{\beta})=\frac{n_{i}}{n}\pi_{ij}(\widehat{\alpha
},\widehat{\beta})$\ is the MLE of the joint probability vector,
$\boldsymbol{p}(\widetilde{\alpha},0)\boldsymbol{=}\left(  p_{11}%
(\widetilde{\alpha},0),p_{12}(\widetilde{\alpha},0),...,p_{I1}%
(\widetilde{\alpha},0),p_{I2}(\widetilde{\alpha},0)\right)  ^{T}$ with
$p_{i1}(\widetilde{\alpha},0)=\frac{n_{i}}{n}\pi_{ij}(\widetilde{\alpha}%
,0)$\ is the MLE of the joint probability vector when the conditional
probabilities are homogeneous\ and%
\[
\mathrm{d}_{Kull}(\boldsymbol{p},\boldsymbol{q})=\sum_{h=1}^{k}p_{h}%
\log\left(  \frac{p_{h}}{q_{h}}\right)  ,
\]
with $\boldsymbol{p}=(p_{1},...,p_{k})^{T}$, $\boldsymbol{q}=(q_{1}%
,...,q_{k})^{T}$ being two arbitrary $k$-dimensional probability vectors. It
is very interesting to observe that $2\mathrm{d}_{Kull}%
(\widehat{\boldsymbol{p}},\boldsymbol{p}(\widetilde{\alpha},0))$, in
(\ref{LRT2}), is the LR test for the homogeneous conditional probabilities
($\beta=0$) and $2\mathrm{d}_{Kull}(\widehat{\boldsymbol{p}},\boldsymbol{p}%
(\widehat{\alpha},\widehat{\beta}))$ is the LR test for the goodness of fit of
the logit model, the test we should perform before the test of monotonicity of probabilities.

The family of test-statistics based on $\phi$-divergence measures, $T_{n,\phi
}(\widehat{\boldsymbol{p}},\boldsymbol{p}(\widetilde{\alpha},0),\boldsymbol{p}%
(\widehat{\alpha},\widehat{\beta})))$, which generalizes the LR test, is
obtained replacing $2$ by $\frac{2}{\phi^{\prime\prime}(1)}$\ and
$\mathrm{d}_{Kull}(\boldsymbol{p},\boldsymbol{q})$ of (\ref{LRT2}) by%
\begin{equation}
\mathrm{d}_{\phi}(\boldsymbol{p},\boldsymbol{q})=\sum_{h=1}^{k}q_{h}%
\phi\left(  \frac{p_{h}}{q_{h}}\right)  , \label{div}%
\end{equation}
where $\phi:%
\mathbb{R}
_{+}\longrightarrow%
\mathbb{R}
$ is a convex function such that $\phi(1)=\phi^{\prime}(1)=0$, $\phi
^{\prime\prime}(1)>0$, $0\phi(\frac{0}{0})=0$, $0\phi(\frac{p}{0}%
)=p\lim_{u\rightarrow\infty}\frac{\phi(u)}{u}$, for $p\neq0$, actually
$\mathrm{d}_{\phi}(\boldsymbol{p},\boldsymbol{q})=\mathrm{d}_{Kull}%
(\boldsymbol{p},\boldsymbol{q})$, where $\phi(x)=x\log x-x+1$. For more
details about $\phi$-divergence measures see Pardo \cite{13}.\ If we take
$\phi_{\lambda}(x)=\frac{1}{\lambda(1+\lambda)}(x^{\lambda+1}-x-\lambda
(x-1))$, where for each $\lambda\in%
\mathbb{R}
-\{-1,0\}$ a different divergence measure is constructed, a very important
subfamily called \textquotedblleft power divergence family of
measures\textquotedblright\ (Cressie and Read \cite{14})\ is obtained%
\begin{equation}
T_{n,\lambda}=2\left(  \mathrm{d}_{\lambda}(\widehat{\boldsymbol{p}%
},\boldsymbol{p}(\widetilde{\alpha},0))-\mathrm{d}_{\lambda}%
(\widehat{\boldsymbol{p}},\boldsymbol{p}(\widehat{\alpha},\widehat{\beta
}))\right)  , \label{CR}%
\end{equation}
where%
\begin{align*}
d_{\lambda}(\boldsymbol{p},\boldsymbol{q})  &  =\frac{1}{\lambda(\lambda
+1)}\left(
{\displaystyle\sum\limits_{h=1}^{k}}
\frac{p_{h}^{\lambda+1}}{q_{h}^{\lambda}}-1\right)  \text{, for each }%
\lambda\in%
\mathbb{R}
-\{-1,0\}\text{,}\\
d_{0}(\boldsymbol{p},\boldsymbol{q})  &  =\lim_{\ell\rightarrow0}d_{\ell
}(\boldsymbol{p},\boldsymbol{q})=\mathrm{d}_{Kull}(\boldsymbol{p}%
,\boldsymbol{q})\text{, for }\lambda=0\text{,}\\
d_{-1}(\boldsymbol{p},\boldsymbol{q})  &  =\lim_{\ell\rightarrow-1}d_{\ell
}(\boldsymbol{p},\boldsymbol{q})=\mathrm{d}_{Kull}(\boldsymbol{q,p})\text{,
for }\lambda=-1\text{.}%
\end{align*}
This family of power divergence based test-statistics includes also the LR
test when $\lambda=0$.

Now we shall establish the distribution of all the test-statistics based on
$\phi$-divergence measures, and thus this distribution is also valid for the
subfamily (\ref{CR}).

\begin{theorem}
\label{Th1}The asymptotic distribution, as $n$ tends to infinite, of the
test-statistics based on $\phi$-divergence measures%
\begin{equation}
T_{n,\phi}=Q_{n,\phi}^{1}(\widehat{\boldsymbol{p}},\boldsymbol{p}%
(\widetilde{\alpha},0))-Q_{n,\phi}^{2}(\widehat{\boldsymbol{p}},\boldsymbol{p}%
(\widehat{\alpha},\widehat{\beta}))), \label{PDT}%
\end{equation}
where%
\begin{align}
Q_{n,\phi}^{1}  &  =\frac{2}{\phi^{\prime\prime}(1)}\mathrm{d}_{\phi
}(\widehat{\boldsymbol{p}},\boldsymbol{p}(\widetilde{\alpha},0)),\nonumber\\
Q_{n,\phi}^{2}  &  =\frac{2}{\phi^{\prime\prime}(1)}\mathrm{d}_{\phi
}(\widehat{\boldsymbol{p}},\boldsymbol{p}(\widehat{\alpha},\widehat{\beta})),
\label{Prelim}%
\end{align}
$\mathrm{d}_{\phi}$ is given by (\ref{div}), is chi-square with one degree of
freedom ($\chi_{1}^{2}$) for the two-sided test (\ref{test2}), and chi-bar
square with two summands ($\frac{1}{2}\chi_{0}^{2}+\frac{1}{2}\chi_{1}^{2}$)
for the one-sided test (\ref{test1}).
\end{theorem}

\begin{proof}
See Appendix \ref{ApA}.
\end{proof}

As noted previously, before performing the test of monotonicity of
probabilities we need to check the goodness of fit of the logit model that we
are considering as assumption. Its test-statistic is the second summand in
(\ref{PDT}) and thus this an advantage in the calculation since we can
calculate both of them at the same time. In the following theorem we give its
asymptotic distribution as preliminary test of (\ref{test1}) or (\ref{test2}).

\begin{theorem}
\label{Th2}The asymptotic distribution, as $n$ tends to infinite, of the
test-statistics based on $\phi$-divergence measures $Q_{n,\phi}^{2}%
(\widehat{\boldsymbol{p}},\boldsymbol{p}(\widetilde{\alpha},\widetilde{\beta
}))$, given in (\ref{Prelim}), is $\chi_{I-2}^{2}$ under the null hypothesis
that the linear logit model is true.
\end{theorem}

\begin{proof}
See Appendix \ref{ApB}.
\end{proof}

\section{Real Data Example\label{Ex}}

Recently, in Paris et al. \cite{16b} dose-response and time-response models
were applied in order to study how some variables influence in two respiratory
diseases, pleural plaques and asbestosis. In total, $n=5545$ formerly
asbestos-exposed workers were considered in a study organized in France from
2003 to 2005. In the original article, four two sided Cochran-Armitage trend
tests were performed by considering four exposures respectively, time since
first exposure (in years), exposure duration (in years), level of exposure
(low, moderate, high and overall, coded by $0$, $1$, $2$ and $3$
respectively), and cumulative exposure index, in relation with the
aforementioned two diseases. The last variable is obtained multiplying the
values of the previous two variables and it can be considered a combination of
them. In this paper, we shall restrict ourselves to two variables, exposure
duration (ED) and cumulative exposure index (CEI). Four exposures were
considered ($I=4$), by splitting the whole interval in four intervals with
around $25\%$ of observed frequencies. In table \ref{t1} the midpoint of each
interval is considered as representative of the interval. In Table \ref{t2},
apart from the one sided CA test-statistic $T_{n,CA}$\ and the two sided one
$T_{n,CA}^{2}$, we studied the family of test-statistics based on power
divergence measures $T_{n,\lambda}=Q_{n,\lambda}^{1}-Q_{n,\lambda}^{2}$, where%
\begin{align*}
Q_{n,\lambda}^{1}  &  =\frac{2}{\lambda(\lambda+1)}\left(  \sum_{i=1}^{4}%
\sum_{j=1}^{2}N_{ij}\left(  \frac{N_{ij}}{n_{i}\pi_{ij}(\widetilde{\alpha}%
,0)}\right)  ^{\lambda}-n\right)  ,\\
Q_{n,\lambda}^{2}  &  =\frac{2}{\lambda(\lambda+1)}\left(  \sum_{i=1}^{4}%
\sum_{j=1}^{2}N_{ij}\left(  \frac{N_{ij}}{n_{i}\pi_{ij}(\widehat{\alpha
},\widehat{\beta})}\right)  ^{\lambda}-n\right)
\end{align*}
with $\lambda\in\{0.5,\frac{2}{3},1,1.5,2\}$and also $G_{n}^{2}=T_{n,0}$\ and
$T_{n,-1}$ where%
\begin{align*}
Q_{n,0}^{1}  &  =2\sum_{i=1}^{4}\sum_{j=1}^{2}N_{i1}\log\left(  \frac{N_{ij}%
}{n_{i}\pi_{ij}(\widetilde{\alpha},0)}\right)  ,\quad Q_{n,0}^{2}=2\sum
_{i=1}^{4}\sum_{j=1}^{2}N_{i1}\log\left(  \frac{N_{ij}}{n_{i}\pi
_{ij}(\widehat{\alpha},\widehat{\beta})}\right)  ,\\
Q_{n,0}^{1}  &  =2\sum_{i=1}^{4}\sum_{j=1}^{2}n_{i}\pi(\widetilde{\alpha
},0)\log\left(  \frac{n_{i}\pi_{ij}(\widetilde{\alpha},0)}{N_{ij}}\right)
,\quad Q_{n,0}^{2}=2\sum_{i=1}^{4}\sum_{j=1}^{2}n_{i}\pi_{ij}%
(\widetilde{\alpha},0)\log\left(  \frac{n_{i}\pi_{ij}(\widehat{\alpha
},\widehat{\beta})}{N_{ij}}\right)  .
\end{align*}
The MLEs of the homogeneous probabilities are $\pi_{i1}(\widetilde{\alpha
},0)=(\sum_{i=1}^{4}N_{i1})/n$, $i=1,...,4$ ($\pi_{i2}(\widetilde{\alpha
},0)=1-\pi_{i1}(\widetilde{\alpha},0)$, $i=1,...,4$), while the monotonic
probabilities are adjusted with a usual binary logistic model. Thus, the
computation for $T_{n,\lambda}$\ is not more complex than for $T_{n,CA}$. The
goodness of fit test-statistics for the linear logit model, $Q_{n,\lambda}%
^{2}$, $\lambda\in\{-1,0.5,0,\frac{2}{3},1,1.5,2\}$, were also calculated. For
all of them the corresponding $p$-value is calculated taking into account that
the asymptotic distribution under the null hypothesis is $T_{n,CA}%
\sim\mathcal{N}(0,1)$ (one sided), $T_{n,\lambda}\sim\frac{1}{2}\chi_{0}%
^{2}+\frac{1}{2}\chi_{1}^{2}$, $T_{n,CA}^{2},T_{n,\lambda}\sim\chi_{1}^{2}$
(two sided), $Q_{n,\lambda}^{2}\sim\chi_{2}^{2}$ (goodness-of-fit).

As two of the four goodness of fit tests reject the hypothesis of linear logit
model, we differ from the conclusion that all trend test were significant.
More thoroughly, we should say it is not possible to consider either
homogeneity or increasing monotonicity in probabilities of pleural plaques in
function of exposure duration (ED), and neither in probabilities of asbestosis
in function of cumulative exposure index (CEI), since the $p$-values of
$Q_{n,\lambda}^{2}$ are very small. On the other hand, the linear logit model
assumption is verified for the other two models (pleural plaques probabilities
in function of ED and asbestosis in function of CEI), since the $p$-values of
$Q_{n,\lambda}^{2}$\ are very large and hence we can perform the test of
monotonicity for their probabilities. From Table \ref{t2} it can be seen that
in case of existing trend in probabilities, we have an increasing trend, since
$\pi_{11}(\widehat{\alpha},\widehat{\beta})<\pi_{21}(\widehat{\alpha
},\widehat{\beta})<\pi_{31}(\widehat{\alpha},\widehat{\beta})<\pi
_{i1}(\widehat{\alpha},\widehat{\beta})$, that is, $\widehat{\beta}>0$, and
hence we could consider the one sided test (\ref{test1}). In view that either
for the one sided or two sided tests we obtain very small $p$-values, the null
hypothesis is rejected and can we conclude that the probability of pleural
plaques increases as exposure index increases, and the probability of
asbestosis increases as the cumulative exposure index increases. It is
remarkable that the obtained $p$-values are in general either or very small or
very big, and this could be motivated by the fact that these conclusions are
obtained with a very large sample size. It is also interesting to mention that
even though two explanatory variables have failed to have a monotonic
influence in probability of disease, we think this is not influenced by the
linear logit link\ \ Even more, both diseases have been proven to increase in
probability when two different explanatory variables are increased.%

\begin{table}[htbp]  \small\tabcolsep3.8pt  \centering
\begin{tabular}
[c]{|ccc|ccc|ccc|}\hline\hline
$i$ & $x_{i}$ & $n_{i}$ & \multicolumn{3}{|c|}{pleural plaques} &
\multicolumn{3}{|c|}{asbestosis}\\\hline\hline
\multicolumn{3}{|c|}{ED} & $n_{i1}$ & $\pi_{i1}(\widetilde{\alpha},0)$ &
$\pi_{i1}(\widehat{\alpha},\widehat{\beta})$ & $n_{i1}$ & $\pi_{i1}%
(\widetilde{\alpha},0)$ & $\pi_{i1}(\widehat{\alpha},\widehat{\beta})$\\\hline
$1$ & \multicolumn{1}{l}{10.0} & 1321 & 179 & 0.1591 & 0.1214 & 71 & 0.0676 &
0.0550\\
$2$ & \multicolumn{1}{l}{24.5} & 1324 & 170 & 0.1591 & 0.1495 & 88 & 0.0676 &
0.0645\\
$3$ & \multicolumn{1}{l}{32.5} & 1408 & 226 & 0.1591 & 0.1673 & 100 & 0.0676 &
0.0704\\
$4$ & \multicolumn{1}{l}{43.0} & 1492 & 307 & 0.1591 & 0.1931 & 116 & 0.0676 &
0.0789\\\hline
\multicolumn{3}{|c|}{CEI} & $n_{i1}$ & $\pi_{i1}(\widetilde{\alpha},0)$ &
$\pi_{i1}(\widehat{\alpha},\widehat{\beta})$ & $n_{i1}$ & $\pi_{i1}%
(\widetilde{\alpha},0)$ & $\pi_{i1}(\widehat{\alpha},\widehat{\beta})$\\\hline
$1$ & 15.0 & 1306 & 150 & 0.1591 & 0.1121 & 50 & 0.0676 & 0.0465\\
$2$ & 41.0 & 1386 & 200 & 0.1591 & 0.1466 & 105 & 0.0676 & 0.0617\\
$3$ & 61.0 & 1380 & 228 & 0.1591 & 0.1692 & 99 & 0.0676 & 0.0720\\
$4$ & 85.0 & 1473 & 304 & 0.1591 & 0.2029 & 121 & 0.0676 &
0.0878\\\hline\hline
\end{tabular}
\caption{Data of the study in Paris et al. (2009) and MLEs of disease proportions.\label{t1}}%
\end{table}%
%

\begin{table}[htbp]  \small\tabcolsep0.8pt  \centering
\begin{tabular}
[c]{llllllll}\hline\hline
\multicolumn{8}{l}{ED vs. pleural plaques}\\\hline
$\lambda$ & $1$ & $-0.5$ & $0$ & $\frac{2}{3}$ & $1$ & $1.5$ & $2$\\
$T_{n,\lambda}$ & $28.2839$ & $28.6098$ & $29.0024$ & $29.6344$ & $29.9992$ &
$30.6104$ & $31.3022$\\
$p\mathrm{-}val(T_{n,\lambda})$ 1s & $5\times10^{-8}$ & $5\times10^{-8}$ &
$4\times10^{-8}$ & $3\times10^{-8}$ & $2\times10^{-8}$ & $2\times10^{-8}$ &
$1\times10^{-8}$\\
$p\mathrm{-}val(T_{n,\lambda})$ 2s & $1\times10^{-7}$ & $9\times10^{-8}$ &
$7\times10^{-8}$ & $5\times10^{-8}$ & $4\times10^{-8}$ & $3\times10^{-8}$ &
$2\times10^{-8}$\\
$Q_{n,\lambda}^{2}$ & $9.3539$ & $9.2922$ & $9.2358$ & $9.1689$ & $9.1389$ &
$9.0981$ & $9.0622$\\
$p\mathrm{-}val(T_{n,\lambda})$ & $0.0093$ & $0.0096$ & $0.0099$ & $0.0010$ &
$0.0010$ & $0.0010$ & $0.0011$\\\hline
$T_{n,CA}$ 1s & $5.3419$ &  &  &  &  &  & \\
$p\mathrm{-}val(T_{n,CA})$ 1s & $5\times10^{-8}$ &  &  &  &  &  & \\
$T_{n,CA}^{2}$ 2s & $28.536$ &  &  &  &  &  & \\
$p\mathrm{-}val(T_{n,CA})$ 2s & $9\times10^{-8}$ &  &  &  &  &  &
\\\hline\hline
\multicolumn{8}{l}{ED vs. asbestosis}\\\hline\hline
$\lambda$ & $1$ & $-0.5$ & $0$ & $\frac{2}{3}$ & $1$ & $1.5$ & $2$\\
$T_{n,\lambda}$ & $6.9869$ & $6.8712$ & $6.7664$ & $6.6430$ & $6.5878$ &
$6.5130$ & $6.4472$\\
$p\mathrm{-}val(T_{n,\lambda})$ 1s & $0.0041$ & $0.0044$ & $0.0046$ & $0.0050$
& $0.0051$ & $0.0053$ & $0.0055$\\
$p\mathrm{-}val(T_{n,\lambda})$ 2s & $0.0082$ & $0.0088$ & $0.0093$ & $0.0099$
& $0.0103$ & $0.0107$ & $0.0111$\\
$Q_{n,\lambda}^{2}$ & $0.1572$ & $0.1573$ & $0.1575$ & $0.1577$ & $0.1578$ &
$0.1580$ & $0.1582$\\
$p\mathrm{-}val(T_{n,\lambda})$ & $0.9244$ & $0.9243$ & $0.9242$ & $0.9242$ &
$0.9241$ & $0.9240$ & $0.9239$\\\hline
$T_{n,CA}$ 1s & $2.5842$ &  &  &  &  &  & \\
$p\mathrm{-}val(T_{n,CA})$ 1s & $0.0049$ &  &  &  &  &  & \\
$T_{n,CA}^{2}$ 2s & $6.6779$ &  &  &  &  &  & \\
$p\mathrm{-}val(T_{n,CA})$ 2s & $0.0098$ &  &  &  &  &  & \\\hline\hline
\multicolumn{8}{l}{CEI vs. pleural plaques}\\\hline\hline
$\lambda$ & $1$ & $-0.5$ & $0$ & $\frac{2}{3}$ & $1$ & $1.5$ & $2$\\
$T_{n,\lambda}$ & $46.5311$ & $46.1887$ & $45.9821$ & $45.9119$ & $45.9632$ &
$46.1473$ & $46.4600$\\
$p\mathrm{-}val(T_{n,\lambda})$ 1s & $<10^{-10}$ & $<10^{-10}$ & $<10^{-10}$ &
$<10^{-10}$ & $<10^{-10}$ & $<10^{-10}$ & $<10^{-10}$\\
$p\mathrm{-}val(T_{n,\lambda})$ 2s & $<10^{-10}$ & $<10^{-10}$ & $<10^{-10}$ &
$<10^{-10}$ & $<10^{-10}$ & $<10^{-10}$ & $<10^{-10}$\\
$Q_{n,\lambda}^{2}$ & $0.4247$ & $0.4246$ & $0.4246$ & $0.4245$ & $0.4245$ &
$0.4244$ & $0.4244$\\
$p\mathrm{-}val(T_{n,\lambda})$ & $0.8087$ & $0.8087$ & $0.8087$ & $0.8088$ &
$0.8088$ & $0.8088$ & $0.8088$\\\hline
$T_{n,CA}$ 1s & $6.7109$ &  &  &  &  &  & \\
$p\mathrm{-}val(T_{n,CA})$ 1s & $<10^{-10}$ &  &  &  &  &  & \\
$T_{n,CA}^{2}$ 2s & $45.0362$ &  &  &  &  &  & \\
$p\mathrm{-}val(T_{n,CA})$ 2s & $<10^{-10}$ &  &  &  &  &  & \\\hline\hline
\multicolumn{8}{l}{CEI vs. asbestosis}\\\hline\hline
$\lambda$ & $1$ & $-0.5$ & $0$ & $\frac{2}{3}$ & $1$ & $1.5$ & $2$\\
$T_{n,\lambda}$ & $24.0869$ & $21.9928$ & $20.1874$ & $18.1522$ & $17.2695$ &
$16.0884$ & $15.0561$\\
$p\mathrm{-}val(T_{n,\lambda})$ 1s & $0.46\times10^{-6}$ & $1.4\times10^{-6}$
& $3.5\times10^{-6}$ & $10.2\times10^{-6}$ & $16.2\times10^{-6}$ &
$30.2\times10^{-6}$ & $52.2\times10^{-6}$\\
$p\mathrm{-}val(T_{n,\lambda})$ 2s & $0.92\times10^{-6}$ & $2.7\times10^{-6}$
& $7.0\times10^{-6}$ & $20.4\times10^{-6}$ & $32.4\times10^{-6}$ &
$60.4\times10^{-6}$ & $104.4\times10^{-6}$\\
$Q_{n,\lambda}^{2}$ & $7.0258$ & $7.0810$ & $7.1465$ & $7.2506$ & $7.3100$ &
$7.4085$ & $7.5187$\\
$p\mathrm{-}val(T_{n,\lambda})$ & $0.0298$ & $0.0290$ & $0.0281$ & $0.0266$ &
$0.0259$ & $0.0246$ & $0.0233$\\\hline
$T_{n,CA}$ 1s & $4.43892$ &  &  &  &  &  & \\
$p\mathrm{-}val(T_{n,CA})$ 1s & $4.5\times10^{-6}$ &  &  &  &  &  & \\
$T_{n,CA}^{2}$ 2s & $19.7040$ &  &  &  &  &  & \\
$p\mathrm{-}val(T_{n,CA})$ 2s & $9.0\times10^{-6}$ &  &  &  &  &  &
\\\hline\hline
\end{tabular}
\caption{One sided and two sided hypothesis testing for monotone proportions in Paris et al. (2009) data.\label{t2}}%
\end{table}%
\newpage

\section{Monte Carlo Study\label{MC}}

Based on a Monte Carlo experiment with 200,000 replications, we compared the
exact type I error probability and power at the $0.05$ nominal significance
level, in order to evaluate the performance of the proposed procedure with the
CA test, within the asymptotic procedures framework. Both versions of the test
for monotonicity of probabilities, one-sided and two sided tests, were taken
into account. We selected as model ED vs. asbestosis from Section \ref{Ex},
that is $I=4$, $x_{1}=10.0$, $x_{2}=24.5$, $x_{3}=32.5$, $x_{4}=43.0$. Since
the sample size is big in the original data set and we are interested in the
performance of small and moderate sample sizes, three scenarios were considered:

\begin{itemize}
\item Scenario 1 (very small sample sizes and balanced): $n_{1}=n_{2}%
=n_{3}=n_{4}=25$.

\item Scenario 2 (small sample sizes and unbalanced): $n_{1}=30$, $n_{1}=40$,
$n_{1}=35$, $n_{1}=25$.

\item Scenario 3 (moderate sample sizes and unbalanced): $n_{1}=50$,
$n_{1}=60$, $n_{1}=55$, $n_{1}=45$.
\end{itemize}

Figures \ref{fig1} and \ref{fig2} show the type I error probabilities of the
tests as a function of the true value of the nuisance parameter when four
test-statistics are considered, $T_{n,\lambda}$, $\lambda\in\{0,\frac{2}%
{3},1\}$, $T_{n,CA}$. We moved 29 values of $\alpha$ until the whole interval
$(0,1)$ was covered for the unknown value of the probabilities, $p=\pi
_{11}(\alpha,0)=\pi_{21}(\alpha,0)=\pi_{31}(\alpha,0)=\pi_{41}(\alpha,0)$. In
Figure \ref{fig2} we can see symmetry with respect to $p=0.5$, actually it is
exactly the same to perform a two hypothesis testing when the true value is
$p_{0}$ or when the true value is $1-p_{0}$, and the role of successful events
and failures can be switched. In Figure \ref{fig1} we cannot see symmetry with
respect to $p=0.5$, and the reason is related to the alternative hypotheses
since the small proportion of samples that we reject tend to verify $\pi
_{11}(\alpha,0)<\pi_{21}(\alpha,0)<\pi_{31}(\alpha,0)<\pi_{41}(\alpha,0)$
seems to be different on the left or right side of $p=0.5$. That is, if the
true value is $p_{0}$ and tends to occur $\pi_{11}(\alpha,0)<\pi_{21}%
(\alpha,0)<\pi_{31}(\alpha,0)<\pi_{41}(\alpha,0)$ then for $1-p_{0}$ tends
occur $1-\pi_{11}(\alpha,0)>1-\pi_{21}(\alpha,0)>1-\pi_{31}(\alpha
,0)>1-\pi_{41}(\alpha,0)$. It is suppose that asymptotically it should not be
difference, but with small and moderate sample it is. For all scenario and for
the two types of contrasts the behavior is quite unstable in the boundaries,
that is when $p$ is close either to $0$ or $1$. For such a case there is a
solution based on the \textquotedblleft pooling design\textquotedblright\ (see
Tebbs and Bilder \cite{14b} for more details) but it goes out from the scope
of the current paper. In Figures \ref{fig1} and \ref{fig2} it is clearly seen
that the LR ($T_{n,0}=G_{n}^{2}$) and CA ($T_{n,CA}$) tests tends to be above
the nominal size but the behavior of the the CA test is much better than the
LR since it remains closer to the nominal size. On the other hand,
$T_{n,\lambda}$, $\lambda\in\{\frac{2}{3},1\}$ tests tends to be below the
nominal size but case $\lambda=\frac{2}{3}$ is usually closer to nominal size
an a little bit flatter. We analyzed also other values of $\lambda$ and we did
not find better choices than $\lambda\in\{\frac{2}{3},1\}$. In Figure
\ref{fig3} the power function the best-test divergence base test statistic in
type I error, $T_{n,2/3}$, and $T_{n,CA}$, are plotted in Scenario 1 (the
power function in the other scenarios are very similar). We can see that the
CA test has in general a little bit higher power than $T_{n,2/3}$, as it
usually happens with test-statistics with higher value of the exact type I
error. Finally, as expected the one sided test has much better power than the
two sided one when $\beta>0$, while when $\beta<0$ the two sided test has
better power. As expected it is concluded that in practice, it is strongly
recommended using the one-sided one for dose-response model when it is logical
to assume that the trend is null or monotonic with a determined direction.%

\begin{figure}[htbp]  \centering
\begin{tabular}
[c]{c}%
\raisebox{0.0934in}{\fbox{\includegraphics[
height=2.5624in,
width=3.8441in
]%
{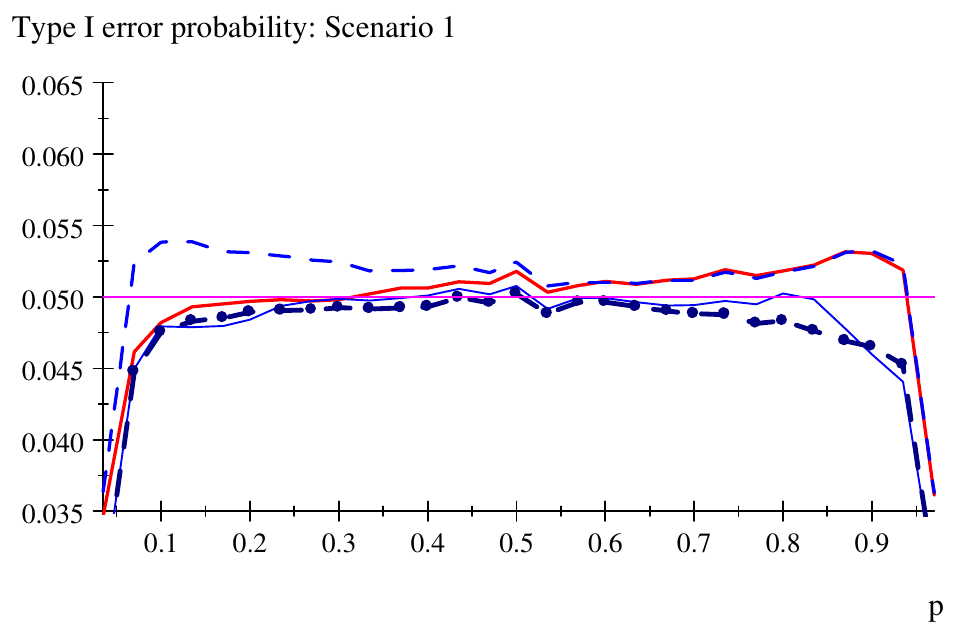}%
}}
\\
\multicolumn{1}{l}{%
{\fbox{\includegraphics[
height=2.5598in,
width=3.845in
]%
{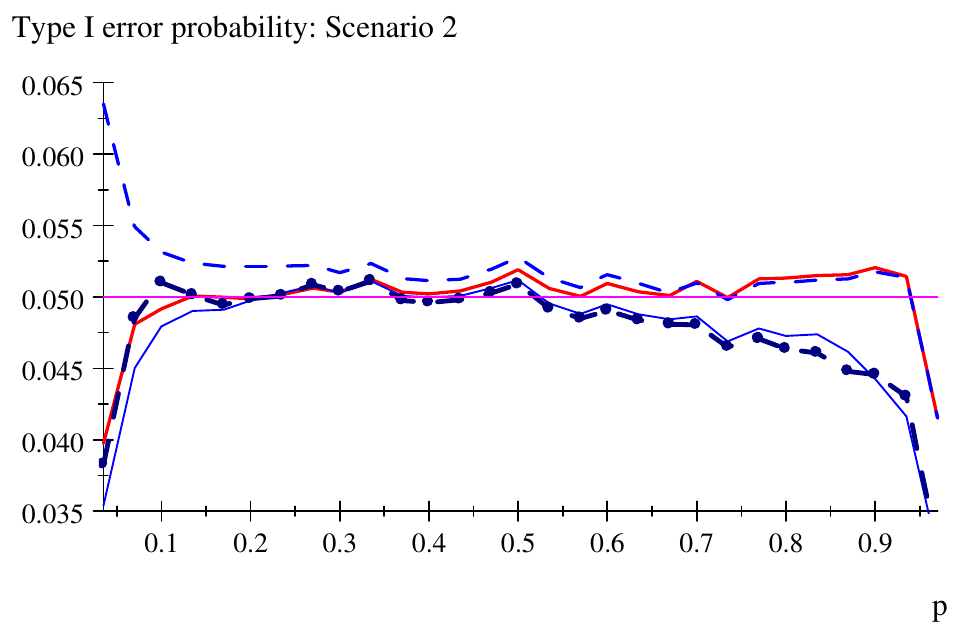}%
}}
\smallskip}\\
\multicolumn{1}{l}{%
{\fbox{\includegraphics[
height=2.5598in,
width=3.845in
]%
{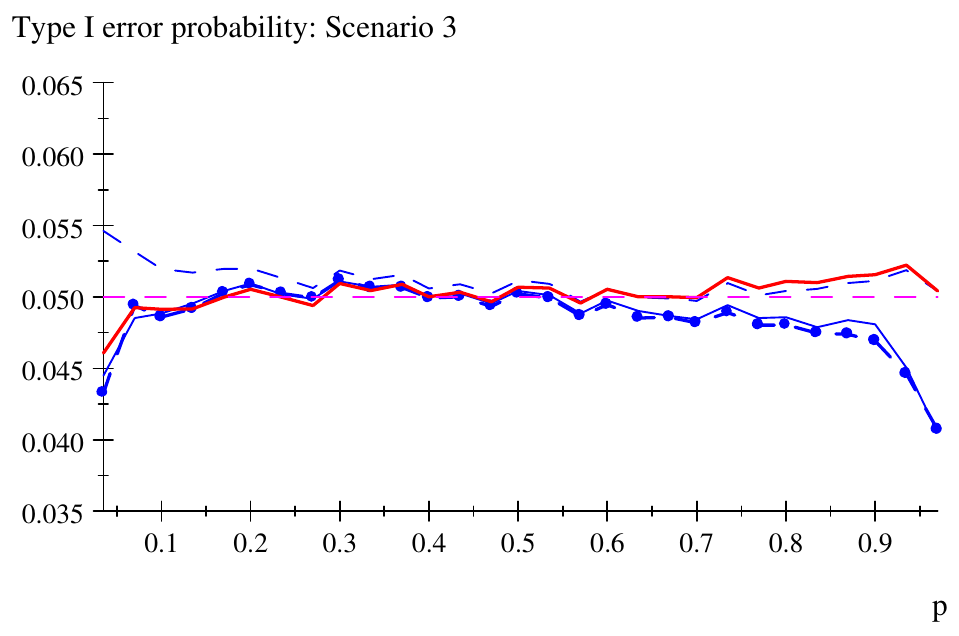}%
}}
}\\%
\raisebox{-0cm}{\includegraphics[
height=0.9797cm,
width=5.1796cm
]%
{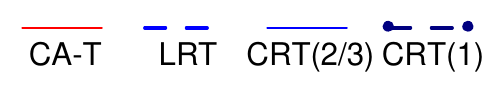}%
}
\end{tabular}
\caption{Exact type I error for one-sided test of trends in probabilities.\label{fig1}}%
\end{figure}%
%

\begin{figure}[htbp]  \centering
\begin{tabular}
[c]{l}%
\raisebox{0.0934in}{\fbox{\includegraphics[
height=2.5624in,
width=3.8441in
]%
{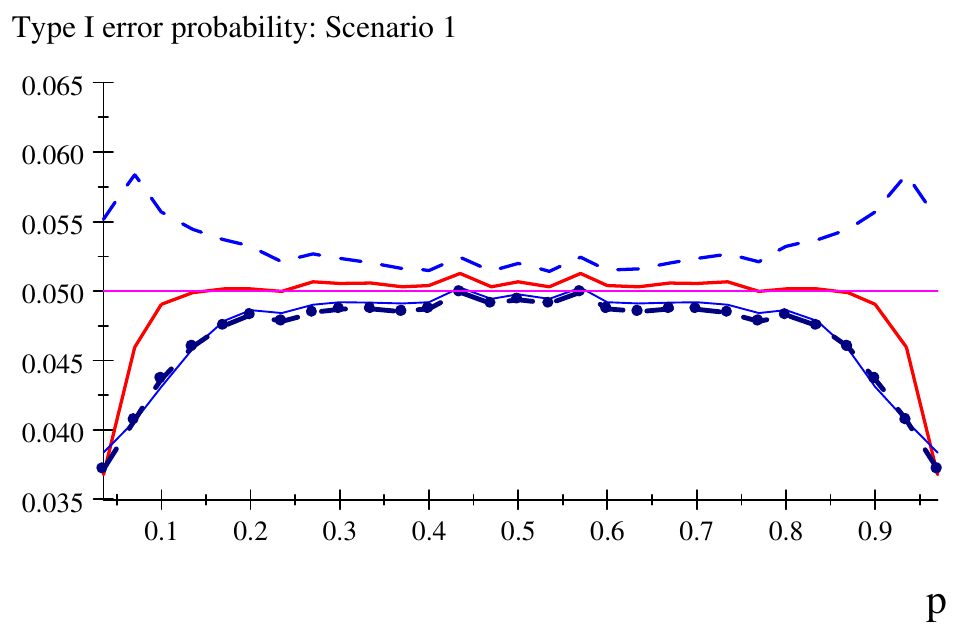}%
}}
\\%
{\fbox{\includegraphics[
height=2.5598in,
width=3.845in
]%
{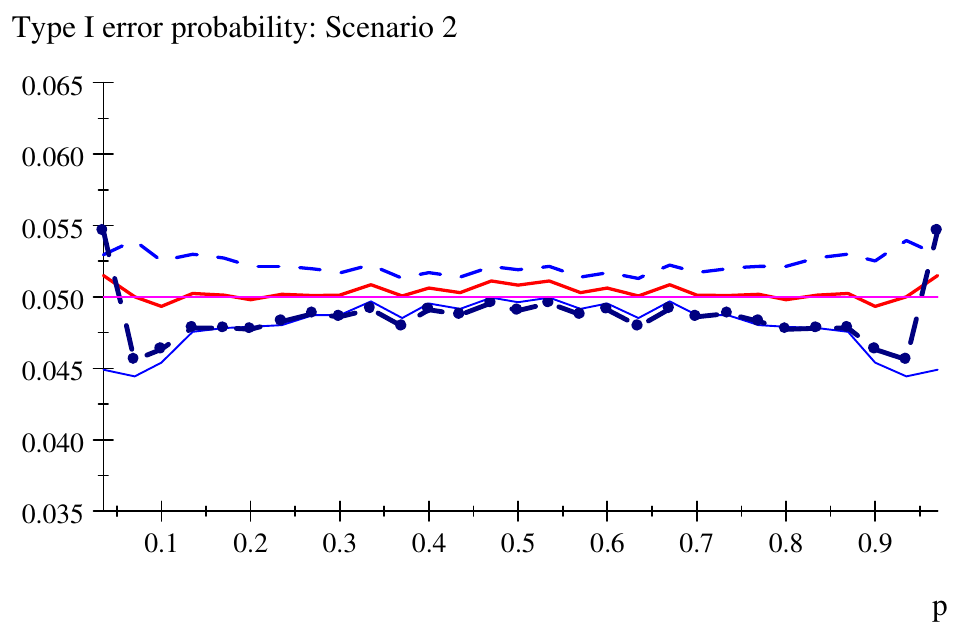}%
}}
\smallskip\\%
{\fbox{\includegraphics[
height=2.5598in,
width=3.845in
]%
{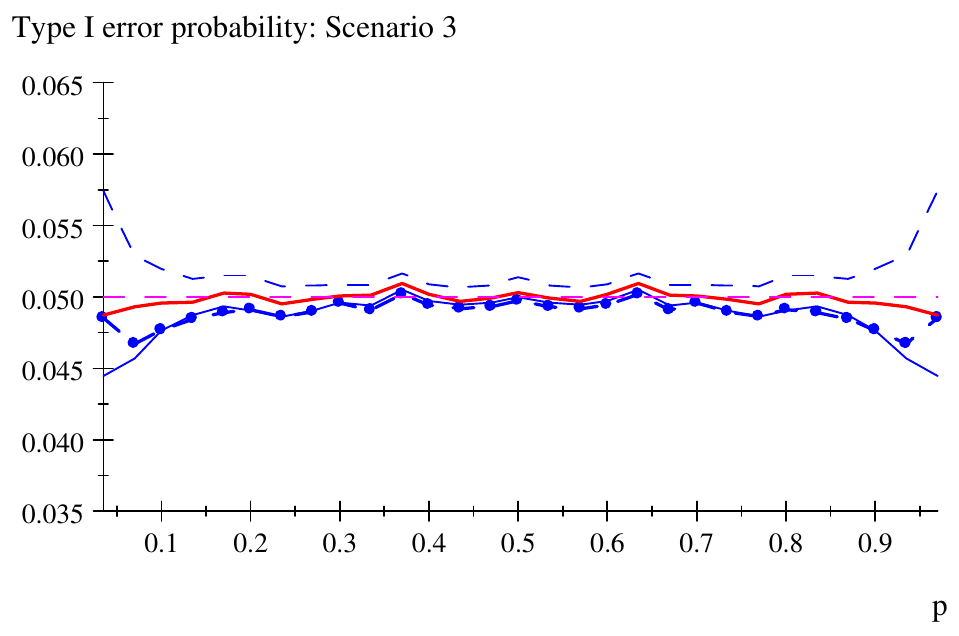}%
}}
\\
\multicolumn{1}{c}{%
\raisebox{-0cm}{\includegraphics[
height=0.9797cm,
width=5.1796cm
]%
{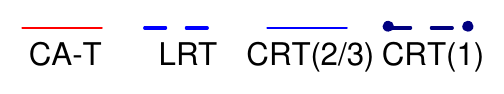}%
}
}%
\end{tabular}
\caption{Exact type I error for two-sided test of trends in probabilities.\label{fig2}}%
\end{figure}%
%

\begin{figure}[htbp]  \centering
\begin{tabular}
[c]{l}%
{\fbox{\includegraphics[
height=2.5598in,
width=3.845in
]%
{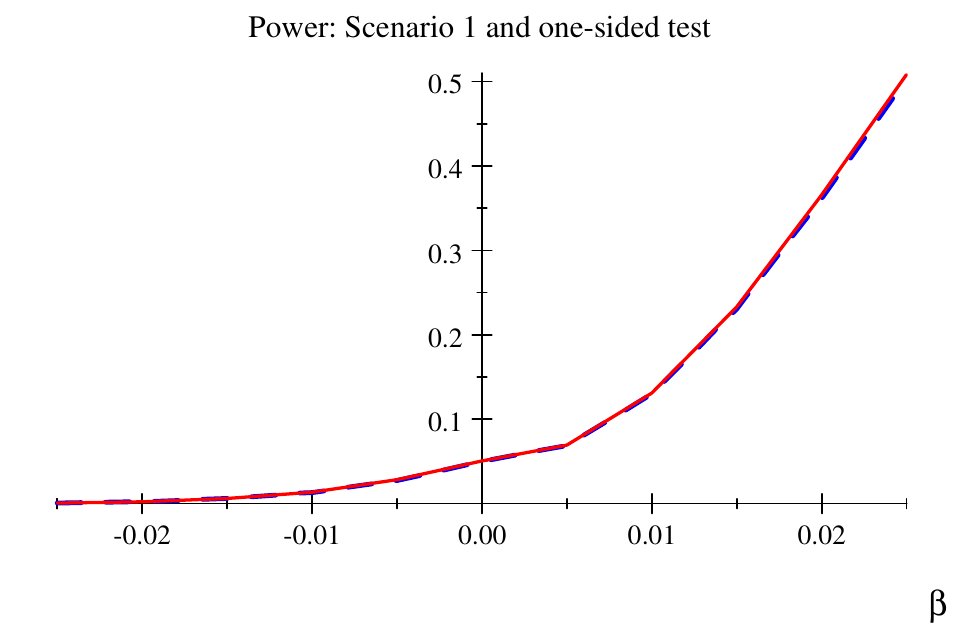}%
}}
\smallskip\\%
{\fbox{\includegraphics[
height=2.5598in,
width=3.845in
]%
{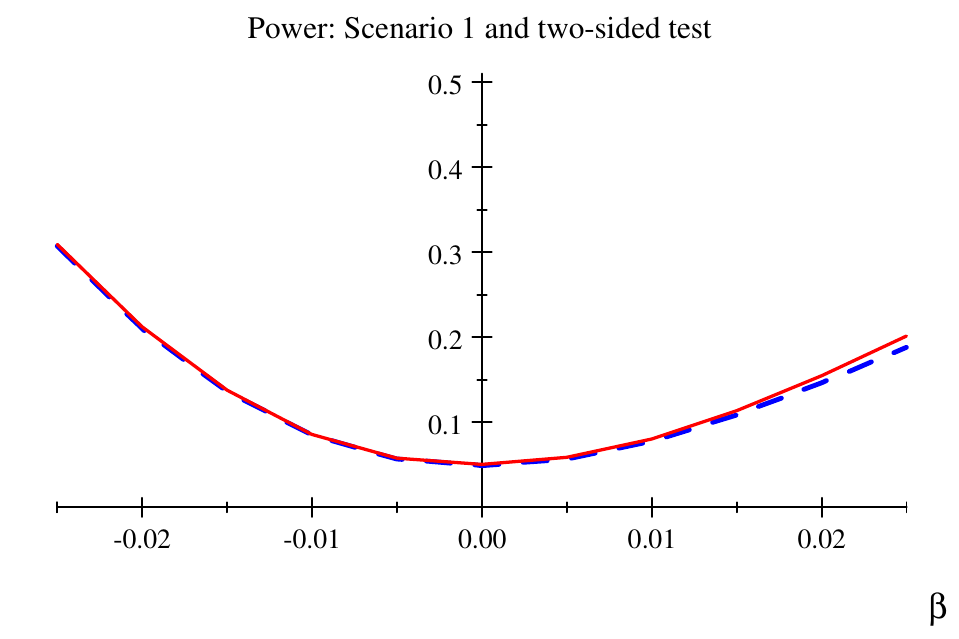}%
}}
\\
\multicolumn{1}{c}{%
\raisebox{-0cm}{\includegraphics[
height=0.9929cm,
width=3.6706cm
]%
{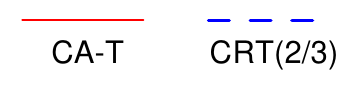}%
}
}%
\end{tabular}
\caption{Exact power for one sided and two sided tests of trends in probabilities.\label{fig2}}%
\end{figure}%

\newpage%

\appendix

\section{Appendix: Proof of Theorem \ref{Th1}\label{ApA}}

In Dardanoni and Forcina \cite{15} generalized models with linear constraints
were used as tool for unifying different kind of order restricted
probabilities when there is no an underlying model. By following a similar
idea, using a saturated loglinear model and in addition, the case of a unique
multi-sample in Mart\'{\i}n and Balakrishnan \cite{16}, it is possible to get
the result we need by considering three models, the saturated model
(nonparametric), linear logit model (adding constraints on the nonparametric
model) and independence model (the loglinear model without interaction
parameter ). The estimated probabilities of these three models are going to be
$\widehat{\boldsymbol{p}}$, $\boldsymbol{p}(\widehat{\alpha},\widehat{\beta}%
)$, $\boldsymbol{p}(\widehat{\alpha},0)$ respectively. Since $\log\frac
{p_{i1}(\alpha,\beta)}{p_{i2}(\alpha,\beta)}=\log\frac{\pi_{i1}(\alpha,\beta
)}{\pi_{i2}(\alpha,\beta)}$, we can express the linear logit model either for
the joint probabilities or conditional probabilities but we shall focus on
joint probabilities%
\[
\log\frac{p_{i1}(\alpha,\beta)}{p_{i2}(\alpha,\beta)}=\alpha+\beta x_{i}.
\]
The joint probabilities in terms of a saturated loglinear model are given by%
\begin{align*}
\log p_{i1}(\alpha,\beta)  &  =\theta_{1(i)}+\theta_{2(1)}+\theta_{12(i1)},\\
\log p_{i2}(\alpha,\beta)  &  =\theta_{1(i)},
\end{align*}
where we have considered the constraints $\theta_{2(2)}=\theta_{12(i2)}=0$ to
avoid overparametrization and without any loss of generality we shall consider
$\theta_{12(I1)}=\frac{x_{I}}{x_{1}}\theta_{12(11)}$. Once we get the values
of $\alpha,\beta$, the terms $\theta_{1(i)}$, $i=1,...,I$, are calculated
taking into account $p_{i1}(\alpha,\beta)+p_{i2}(\alpha,\beta)=\frac{n_{i}}%
{n}$, $i=1,...,I$. If we take the ratio of both logarithm of probabilities we
have%
\[
\log\frac{p_{i1}(\alpha,\beta)}{p_{i2}(\alpha,\beta)}=\theta_{2(1)}%
+\theta_{12(i1)},
\]
which means that $\theta_{2(1)}=\alpha$, $\theta_{12(i1)}=\beta x_{i}$, and
thus the linear logit model can be reparametrized as a saturated loglinear
model subject to the linear constraint
\begin{equation}
x_{1}\theta_{12(i1)}-x_{i}\theta_{12(11)}=0,\quad i=2,...,I-1 \label{res}%
\end{equation}
(the equation is also true for $i=I$ but it was true for the saturated model).
In matrix notation the saturated loglinear model is given by%
\[
\log\boldsymbol{p}(\boldsymbol{\theta})=\boldsymbol{W}_{1}\boldsymbol{\theta
}_{1}+\boldsymbol{W}_{12}\boldsymbol{\theta}_{12}+\boldsymbol{w}_{2}%
\theta_{2(1)}=\boldsymbol{W}_{1}\boldsymbol{\theta}_{1}+\boldsymbol{W\theta},
\]
where $\boldsymbol{\theta}_{1}=(\theta_{1(1)},...,\theta_{1(I-1)})^{T}$,
$\boldsymbol{\theta}_{12}=(\theta_{12(11)},...,\theta_{12(I-1,1)})^{T}$,
$\boldsymbol{\theta}=(\boldsymbol{\theta}_{12}^{T},\theta_{2(1)})^{T}$,%
\[
\boldsymbol{W}_{1}=\boldsymbol{I}_{I}\otimes\boldsymbol{1}_{2},\quad
\boldsymbol{W}_{12}=%
\begin{pmatrix}
\boldsymbol{I}_{I-1}\\
\frac{x_{I}}{x_{1}}\boldsymbol{e}_{1}^{T}%
\end{pmatrix}
\otimes%
\begin{pmatrix}
1\\
0
\end{pmatrix}
,\quad\boldsymbol{w}_{2}=\boldsymbol{1}_{I}\otimes%
\begin{pmatrix}
1\\
0
\end{pmatrix}
,\quad\boldsymbol{W}=(\boldsymbol{W}_{12},\boldsymbol{w}_{2}),
\]
$\otimes$ is the Kronecker product (see Chapter 16 of Harville \cite{17}),
$\boldsymbol{I}_{a}$\ is the the identity matrix of order $a$, $\boldsymbol{e}%
_{i}$\ is the vector of zeros and $1$ in the $i$-th position and
$\theta_{1(i)}=\log(\frac{n_{i}}{n})-\log(\boldsymbol{1}_{2}^{T}%
\exp\{(\boldsymbol{I}_{2}\otimes\boldsymbol{e}_{i}^{T})\boldsymbol{W\theta
}\})$. The last expression is similar to the formula for getting the intercept
in a product-multinomial sampling. Condition (\ref{res}) in matrix notation is
given by $(x_{1}\boldsymbol{1}_{I-2},-diag\{x_{i}\}_{i=2}^{I-1},\boldsymbol{0}%
_{I-2})\boldsymbol{\theta}=\boldsymbol{0}_{I-2}$. In this framework, for the
linear logit model, (\ref{test1}) is equal to%
\begin{align*}
H_{0}  &  :\theta_{12(11)}=0\text{, }x_{1}\theta_{12(i1)}-x_{i}\theta
_{12(11)}=0,\quad i=2,...,I-1,\\
H_{1}  &  :\theta_{12(11)}>0\text{, }x_{1}\theta_{12(i1)}-x_{i}\theta
_{12(11)}=0,\quad i=2,...,I-1,
\end{align*}
for the saturated loglinear model. For the one-sided test we have three
parametric spaces%
\begin{align*}
\Omega(E)  &  =\{\boldsymbol{\theta}\in%
\mathbb{R}
^{I}:\boldsymbol{e}_{1}^{T}\boldsymbol{\theta}\leq0,\quad(x_{1}\boldsymbol{1}%
_{I-2},-diag\{x_{i}\}_{i=2}^{I-1},\boldsymbol{0}_{I-2})\boldsymbol{\theta
}=\boldsymbol{0}_{I-2}\},\\
\Theta(F)  &  =\{\boldsymbol{\theta}\in%
\mathbb{R}
^{I}:(x_{1}\boldsymbol{1}_{I-2},-diag\{x_{i}\}_{i=2}^{I-1},\boldsymbol{0}%
_{I-2})\boldsymbol{\theta}=\boldsymbol{0}_{I-2}\},\\
\Theta &  =\Theta(\varnothing)=%
\mathbb{R}
^{I},
\end{align*}
such that $\Theta(E)\subset\Omega(F)\subset\Theta$, this statistical problem
can be placed in the nesting framework of the paper Mart\'{\i}n and
Balakrishnan \cite{16}. In terms of the hypothesis testing formulation given
in Mart\'{\i}n and Balakrishnan \cite[Section 2]{16}, the one sided hypothesis
testing $H_{0}$: $\boldsymbol{\theta}\in\Omega(E)$ vs. $H_{1}$:
$\boldsymbol{\theta}\in\Theta(F)-\Omega(E)$ is (12), the set of indices that
the restriction is active is the same for for the null and alternative
hypothesis, $E=F=\{i\in\{2,...,I-1\}:h_{i}(\boldsymbol{\theta})=0\}$, with
$h_{i}(\boldsymbol{\theta})=x_{1}\theta_{12(i1)}-x_{i}\theta_{12(11)}$,
$i=2,...,I-1$. The LR test match formula (20) in Mart\'{\i}n and Balakrishnan
\cite[Section 2]{16} and this is a particular test-statistics of the second
test-statistic given in Definition 16 for which the same idea of (54) in the
simulation study is used. Hence, the asymptotic distribution for one-sided
test (\ref{test1}) is obtained from Theorem 17 in Mart\'{\i}n and Balakrishnan
\cite[Section 2]{16}. The maximum number of positions in parameter $\beta$
where $\beta=0$ is reached, is $1$ (if $\widehat{\beta}>0$, then
$\widetilde{\beta}=0$), which means that the chi-bar square distribution of
test (\ref{test1}) has two summands with weights equals $\frac{1}{2}$. For the
two-sided test we have three parametric spaces%
\begin{align*}
\Theta(E^{\prime})  &  =\{\boldsymbol{\theta}\in%
\mathbb{R}
^{I}:\boldsymbol{e}_{1}^{T}\boldsymbol{\theta}=0,\quad(x_{1}\boldsymbol{1}%
_{I-2},-diag\{x_{i}\}_{i=2}^{I-1},\boldsymbol{0}_{I-2})\boldsymbol{\theta
}=\boldsymbol{0}_{I-2}\},\\
\Theta(F^{\prime})  &  =\{\boldsymbol{\theta}\in%
\mathbb{R}
^{I}:(x_{1}\boldsymbol{1}_{I-2},-diag\{x_{i}\}_{i=2}^{I-1},\boldsymbol{0}%
_{I-2})\boldsymbol{\theta}=\boldsymbol{0}_{I-2}\},\\
\Theta &  =\Theta(\varnothing)=%
\mathbb{R}
^{I},
\end{align*}
such that $\Theta(E^{\prime})\subset\Theta(F^{\prime})\subset\Theta$, this
statistical problem can be placed in the nesting framework of the paper
Mart\'{\i}n and Balakrishnan \cite{16}. In terms of the hypothesis testing
formulation given in Mart\'{\i}n and Balakrishnan \cite[Section 2]{16}, the
two sided hypothesis testing $H_{0}^{\prime}$: $\boldsymbol{\theta}\in
\Theta(E^{\prime})$ vs. $H_{1}^{\prime}$: $\boldsymbol{\theta}\in
\Theta(F^{\prime})-\Theta(E^{\prime})$ is (10), the set of indices that the
restriction is active for the null hypothesis is $E^{\prime}=\{i\in
\{1,2,...,I-1\}:h_{i}(\boldsymbol{\theta})=0\}$, with $h_{1}%
(\boldsymbol{\theta})=\boldsymbol{e}_{1}^{T}\boldsymbol{\theta}$,
$h_{i}(\boldsymbol{\theta})=x_{1}\theta_{12(i1)}-x_{i}\theta_{12(11)}$,
$i=2,...,I-1$, and $F^{\prime}=F$ for the alternative hypothesis. The LR test
match formula (18) in Mart\'{\i}n and Balakrishnan \cite[Section 2]{16} and
this is a particular test-statistics of the second test-statistic given in
Definition 7. Hence, the asymptotic distribution for two-sided test
(\ref{test2}) is obtained from Theorem 8 in Mart\'{\i}n and Balakrishnan
\cite[Section 2]{16}. Note that $\mathrm{card}(E^{\prime})-\mathrm{card}%
(F^{\prime})=1$,\ which means that the chi-square distribution of
(\ref{test2})\ has one degree of freedom.

\section{Appendix: Proof of Theorem \ref{Th2}\label{ApB}}

We can follow the same idea of the previous proof. The asymptotic
distributions is obtained from Theorem 8 in Mart\'{\i}n and Balakrishnan
\cite[Section 2]{16}. The parametric spaces are%
\begin{align*}
\Theta(E^{\prime\prime})  &  =\{\boldsymbol{\theta}\in%
\mathbb{R}
^{I}:(x_{1}\boldsymbol{1}_{I-2},-diag\{x_{i}\}_{i=2}^{I-1},\boldsymbol{0}%
_{I-2})\boldsymbol{\theta}=\boldsymbol{0}_{I-2}\},\\
\Theta(F^{\prime\prime})  &  =\Theta=\Theta(\varnothing)=%
\mathbb{R}
^{I},
\end{align*}
such that $\Theta(E^{\prime\prime})\subset\Theta(F^{\prime\prime})=\Theta$. In
terms of the hypothesis testing formulation given in Mart\'{\i}n and
Balakrishnan \cite[Section 2]{16}, the goodness of fit hypothesis testing
$H_{0}$: $\boldsymbol{\theta}\in\Theta(E^{\prime\prime})$ vs. $H_{1}$:
$\boldsymbol{\theta}\in\Theta(F^{\prime\prime})-\Theta(E^{\prime\prime})$ is
(10), the set of indices that the restriction is active for the null
hypothesis is $E^{\prime\prime}=E=F$ and $F^{\prime\prime}=\varnothing$ (the
saturated model does not considers constraints)\ for the alternative
hypothesis. The LR test match formula (18) in Mart\'{\i}n and Balakrishnan
\cite[Section 2]{16} and this is a particular test-statistics of the second
test-statistic given in Definition 7. Hence, the asymptotic distribution for
two-sided test (\ref{test2}) is obtained from Theorem 8 in Mart\'{\i}n and
Balakrishnan \cite[Section 2]{16}. Note that $\mathrm{card}(E^{\prime\prime
})-\mathrm{card}(F^{\prime\prime})=I-2$,\ which means that the chi-square
distribution has $I-2$ degrees of freedom under the hypothesis that the linear
logit model is true.


\newpage

\begin{thebibliography}{99}                                                                                               %


\bibitem {1}Cochran WG. Some Methods for Strengthening the Common $\chi^{2}$
Tests. \emph{Biometrics} 1954; 10: 417-451.

\bibitem {2}Armitage P. Tests for Linear Trends in Proportions and
Frequencies. \emph{Biometrics }1955; 11: 375-386.

\bibitem {3}Mantel N. Chi-square Tests with one degree of freedom extensions
of the mantel-Haenszel procedure. \emph{Journal of the American Statistical
Association}\textit{ 1963; }58: 690-700.

\bibitem {4}Tarone RE and Gart JJ. On the Robustness of Combined Test for
Trends in Proportions. \emph{Journal of the American Statistical Association,}
1980; 75: 110-116.

\bibitem {5}Rao CR. \emph{Linear Statistical Inference and Its Applications}.
New York: Wiley, 1973.

\bibitem {6}Cox DR. \emph{Analysis of Binary Data}. London: Methuen, 1970.

\bibitem {6b}Cox DR. Note on Grouping. \emph{Journal of the American
Statistical Association} 1957; 52: 543-547.

\bibitem {6c}Barlow RE, Bartholomew DJ, Bremmer JM and Brunk, HD.
\emph{Statistical Inference Under Order Restrictions}\textit{. }New York: John
Wiley \& Sons, 1972.

\bibitem {6d}Robertson T, Wright FT and Dykstra RL. \emph{Order Restricted
Statistical Inference}\textit{. }New York:John Wiley \& Sons, 1988.

\bibitem {6e}Silvapulle MJ and Sen PK. \emph{Constrained Statistical
Inference: Order, Inequality, and Shape Constraints}. New York: Wiley Series
in Probability and Statistics, 2004.

\bibitem {7}Leuraud K and Benichou J. A comparison of several methods to test
for the existence of a monotonic dose-response relationship in clinical and
epidemiological studies. \emph{Statistics in Medicine }2001;\textit{
}20\textbf{: }3335-3351.

\bibitem {8}Agresti A and Coull BA. An empirical comparison of inference using
a order-restricted and linear logit models for a binary response.
\emph{Communications in Statistics (Simulation)}\textit{ 1998; }27: 147-166.

\bibitem {9}Hirji KF and Tang ML. A comparison of Tests for Trend.
\emph{Communications in Statistics (Theory and Methods)} 1998; 27: 943-963.

\bibitem {10}Tang ML, Chan PS and Chan W. On Exact Unconditional Test for
Linear Trend in Dose-Response Studies. \emph{Biometrical Journal} 2000; 42: 795-806.

\bibitem {11}Shan G, Ma C. and Wilding GE. An Efficient and Exact Approach for
Detecting Trends with Binary Endpoints. \emph{Statistics in Medicine} 2012;
31: 155-164.

\bibitem {12}Kang S and Lee J. The size of the Cochran-Armitage trend test in
2xC contingency tables. \emph{Journal of Statistical Planning and Inference}
2007; 137: 1851-1861.

\bibitem {13}Pardo, L. \emph{Statistical Inference Based on Divergence
Measures}\textit{. }New York: Chapman \& Hall/CRC, 2006.

\bibitem {14}Cressie N and Read TRC: Multinomial goodness-of-fit tests.
\emph{Journal of the Royal Statistical Society, Series B} 1984; 46\textbf{:}440-464.

\bibitem {14b}Tebbs JM and Bilder CR. Hypotesis Tests for and against a Simple
Order among Proportions Estimated by Pooled Testing. \emph{Biometrical
Journal} 2006; 48: 792-804.

\bibitem {16}Mart\'{\i}n N and Balakrishnan N. Hypothesis testing in a generic
nesting framework with general population distributions. \emph{Journal of
Multivariate Analysis} 2013; 118: 1-23.

\bibitem {16b}Paris C, Thierry S, Brochard P, Letourneux M, Schorle E,
Stoufflet A, Ameille J, Conso F and Pairon JC. Pleural plaques and asbestosis:
dose- and time-response relationships based on HRCT data. \emph{European
Respiratory Journal} 2009; 34: 72-79.

\bibitem {15}Dardanoni V and Forcina A. A Unified Approach to Likelihood
Inference on Stochastic Orderings in a Nonparametric Context. \emph{Journal of
the American Statistical Association} 1998; 93: 1112-1122.

\bibitem {17}Harville DA. \emph{Matrix Algebra From a Statistician's
Perspective}. New York: Springer, 2008.
\end{thebibliography}
\end{document}